\documentclass[11pt]{article}
\usepackage[margin=1in]{geometry} 
\usepackage{amsmath}
\usepackage{amssymb,amsfonts,bbm,mathrsfs,stmaryrd}

\usepackage[colorlinks,
             linkcolor=black!50!red,
             citecolor=blue,
             pdftitle={On the coordinate (in)dependence of the formal path integral},
             pdfauthor={Theo Johnson-Freyd},
             pdfsubject={quantum mechanics},
             pdfkeywords={path integral, quantum mechanics},
             pdfproducer={pdfLaTeX},
             pdfpagemode=None,
             bookmarksopen=true
             bookmarksnumbered=true]{hyperref}

\usepackage[amsmath,thmmarks,hyperref]{ntheorem}
\usepackage{cleveref}

\theoremstyle{change}

\newtheorem{defn}[equation]{Definition}
\newtheorem{thm}[equation]{Theorem}

\newtheorem{lemma}[equation]{Lemma}

\theorembodyfont{\upshape}
\theoremsymbol{\ensuremath{\Diamond}}

\theoremstyle{nonumberplain}

\theoremsymbol{\ensuremath{\Box}}
\newtheorem{proof}{Proof}

\qedsymbol{\ensuremath{\Box}}

\creflabelformat{equation}{#2(#1)#3} 

\crefname{equation}{equation}{equations}
\crefname{eg}{example}{examples}
\crefname{defn}{definition}{definitions}
\crefname{prop}{proposition}{propositions}
\crefname{thm}{Theorem}{Theorems}
\crefname{lemma}{lemma}{lemmas}
\crefname{cor}{corollary}{corollaries}
\crefname{remark}{remark}{remarks}
\crefname{section}{Section}{Sections}
\crefname{subsection}{Section}{Sections}

\crefformat{equation}{#2equation~(#1)#3} 
\crefformat{eg}{#2example~#1#3} 
\crefformat{defn}{#2definition~#1#3} 
\crefformat{prop}{#2proposition~#1#3} 
\crefformat{thm}{#2Theorem~#1#3} 
\crefformat{lemma}{#2lemma~#1#3} 
\crefformat{cor}{#2corollary~#1#3} 
\crefformat{remark}{#2remark~#1#3} 
\crefformat{section}{#2Section~#1#3} 
\crefformat{subsection}{#2Section~#1#3} 

\Crefformat{equation}{#2Equation~(#1)#3} 
\Crefformat{eg}{#2Example~#1#3} 
\Crefformat{defn}{#2Definition~#1#3} 
\Crefformat{prop}{#2Proposition~#1#3} 
\Crefformat{thm}{#2Theorem~#1#3} 
\Crefformat{lemma}{#2Lemma~#1#3} 
\Crefformat{cor}{#2Corollary~#1#3} 
\Crefformat{remark}{#2Remark~#1#3} 
\Crefformat{section}{#2Section~#1#3} 
\Crefformat{subsection}{#2Section~#1#3}

\numberwithin{equation}{section}

\usepackage{tikz}
\tikzset{dot/.style={circle,draw,fill,inner sep=1pt}}

\newcommand\Aa{{\cal A}}
\newcommand\Oo{{\cal O}}
\newcommand\Uu{{\cal U}}
\newcommand\NN{{\mathbb N}}
\newcommand\RR{{\mathbb R}}
\newcommand\Ddd{\mathscr{D}}
\renewcommand{\d}{{\,\rm d}}
\newcommand\T{{\rm T}}

\newcommand\mono{\hookrightarrow}
\newcommand\sminus{\smallsetminus}
\newcommand\st{{\textrm{ s.t.\ }}}

\DeclareMathOperator{\Aut}{Aut}
\DeclareMathOperator{\dVol}{dVol}
\DeclareMathOperator{\ev}{ev}
\DeclareMathOperator{\fiber}{fiber}
\DeclareMathOperator{\GL}{GL}
\DeclareMathOperator{\id}{id}
\DeclareMathOperator{\sign}{sign}
\DeclareMathOperator{\tr}{tr}

\newcommand{\define}[1]{{\bf #1}}

\title{On the coordinate (in)dependence of the formal path integral}
\author{Theo Johnson-Freyd}

\begin{document}
\maketitle

\begin{abstract}
  When path integrals are discussed in quantum field theory, it is almost always assumed that the fields take values in a vector bundle.  When the fields are instead valued in a possibly-curved fiber bundle, the independence of the formal path integral on the coordinates becomes much less obvious.  In this short note, aimed primarily at mathematicians, we first briefly recall the notions of Lagrangian classical and quantum field theory and the standard coordinate-full definition of the ``formal'' or ``Feynman-diagrammatic'' path integral construction.  We then outline a proof of the following claim: the formal path integral does not depend on the choice of coordinates, but only on a choice of fiberwise volume form.  Our outline is an honest proof when the formal path integral is defined without ultraviolet divergences.
\end{abstract}

\section{Introduction}

Feynman's path integral construction has been an extremely valuable tool in physics and mathematics \cite{Witten1989,KolyaShort}: given a classical field theory, it defines a quantum field theory in terms of a usually ill-defined infinite-dimensional integral.  Because analytic definitions of this integral have proven difficult, the definition that has been the most used in physics and that in mathematics (e.g.\ \cite{BN1995}) is of the ``formal'' or ``Feynman-diagrammatic'' path integral.

Classical field theory makes sense on any smooth fiber bundle, but the literature on quantum field theory almost always assumes that the fibers are affine spaces.  In the quantization of such ``linear'' field theories, the analytically problematic choice of measure of integration can be safely ignored, as finite-dimensional vector space have unique translation-invariant measures, up to a scalar.  However, when quantizing non-linear field theories, one must choose a measure, and it is not clear that the formal path integral is constant under volume-preserving changes of coordinates.

In \cref{section1} we very briefly describe classical and quantum Lagrangian field theory, and state our main theorem: the formal path integral does not depend on the choice of coordinates, but only on a choice of fiberwise volume form on the corresponding bundle.  In \cref{section2} we review the Feynman-diagrammatic description of the asymptotics of finite-dimensional oscillating integrals and the corresponding definition of formal integration.  The proof of the main result comprises \cref{section3,section4}.
Our argument is essentially formal, and so applies whenever the Feynman diagrams are defined and satisfy certain natural identities.

\section{Review of Lagrangian field theory, and the main theorem} \label{section1}

Classical Lagrangian field theory, in its simplest case, begins with a smooth bundle $\pi: P\to X$, where $X$ is a finite-dimensional manifold with boundary $\partial X$.  Let $\Gamma(P\to X)$ be the set of smooth\footnote{Mutatis mutandis, one can consider sections of different regularity, or let $P\to X$ be an arbitrary submersion.} sections of $P \to X$.  For each \define{Dirichlet boundary condition} $\phi \in \Gamma(\partial P \to \partial X)$, let $\Gamma_\phi(P\to X) = \{\gamma \in\Gamma(P\to X) \st \gamma|_{\partial X} = \phi\}$.  Let $\T^{\fiber} P = \ker\{ \! \d \pi: \T P \to \T X\}$.
Then $\Gamma(P\to X)$ and $\Gamma_\phi(P\to X)$ are infinite-dimensional smooth manifolds: the tangent space $\T_\gamma\bigl(\Gamma(P \to X)\bigr)$ is the space of smooth sections of the pullback vector bundle $\gamma^*\bigl(\T^{\fiber}P\bigr) \to X$, and
$\T_\gamma\bigl(\Gamma_\phi(P \to X)\bigr) = \bigl\{ \xi \in \T_\gamma\bigl(\Gamma(P \to X)\bigr) \st \xi|_{\partial X} = 0\bigr\}$.
An \define{action} on $P \to X$ is a smooth function $\Aa: \Gamma(P\to X) \to \RR$.
The \define{classical Lagrangian field theory} for the bundle $P\to X$ and action $\Aa$
assigns to each $\phi \in \Gamma(\partial P \to \partial X)$ the {\em set} of critical points of the restriction to $\Gamma_{\phi}(P\to X)$ of the action $\Aa$.

Quantum Lagrangian field theory requires in addition to the inputs $P\to X$ and $\Aa: \Gamma(P \to X) \to \RR$ one more piece of classical data.  Recall that for a finite-dimensional manifold $M$, its \define{density bundle} is the trivial one-dimensional vector bundle $\Ddd_M \to M$ given in local coordinates $m,\tilde m: M \to \RR^{\dim M}$ by the transition functions $\tilde \mu = \bigl| \det \frac{\partial m}{\partial \tilde m}\bigr| \mu$.  (When $M$ is oriented, $\Ddd_M = \bigwedge^{\text{top}}\T^*M$.)  Given a bundle $\pi: P \to X$, its \define{fiberwise density bundle} $\Ddd_{\fiber}\to P$ is the bundle that restricts to the density bundles $\Ddd_{\pi^{-1}(x)} \to \pi^{-1}(x)$ on the fibers ($\Ddd_{\fiber} = \bigwedge^{\dim P - \dim X}\T^{\fiber} P$ when the fibers are oriented).  The additional datum for the quantum theory is a smooth section $\dVol \in \Gamma(\Ddd_{\fiber} \to P)$.  Then Feynman \cite{FeynmanHibbs1965} says that the quantum theory should assign to each Dirichlet boundary condition $\phi$ and each $\hbar \in \RR_{\neq 0}$ the following {\em complex number}:
\begin{equation} \label{FeynmanIntegral}
  \Uu^\hbar(\phi) = \int_{\gamma \in \Gamma_\phi(P\to X)} \exp\biggl( \frac i \hbar \,\Aa(\gamma)\biggr)\prod_{x\in X \sminus \partial X} \dVol\bigl(\gamma(x)\bigr)
\end{equation}

The integral in \cref{FeynmanIntegral} is analytically ill-defined: $\Gamma_\phi(P \to X)$ is infinite-dimensional unless $(\dim X)(\dim P - \dim X) = 0$, and the ``measure'' $\d\gamma = \prod_{x\in X \sminus \partial X} \dVol\bigl(\gamma(x)\bigr)$ is an infinite product.  The usual approach to algebraically defining this integral follows the following steps:
\begin{enumerate}
  \item Restrict the goal to a definition when $\hbar$ is a formal variable.  Impose the stationary-phase approximation that as $\hbar \to 0$, the integral should be supported near critical points of $\Aa|_{\Gamma_\phi(P\to X)}$.  Pick one such critical point $\gamma$, and suppose that the Hessian of $\Aa|_{\Gamma_\phi(P\to X)}$ at $\gamma$, a map $\Aa^{(2)}: \T_\gamma\bigl(\Gamma_\phi(P\to X)\bigr) \to \T^*_\gamma\bigl(\Gamma_\phi(P\to X)\bigr)$, has trivial kernel. 
  \item Choose ``coordinates near $\gamma$,'' i.e.\ an identification of bundles over $X$ between some open subbundle of $\gamma^*\bigl(\T^{\fiber}P\bigr)\to X$ and an open neighborhood of the image of $\gamma$ in $P \to X$.  Using \cite{Moser1965}, choose these coordinates to be \define{volume preserving}: the fiberwise volume form $\dVol \in \Gamma( \Ddd_{\fiber} \to P)$ determines a translation-invariant fiberwise volume form $\dVol(\gamma)$ on $\gamma^*\bigl(\T^{\fiber}P\bigr)$, and these should agree under the identification.
  \item Copy the usual Feynman-diagrammatic description of the asymptotics of oscillating integrals over finite-dimensional vector spaces (see \cref{section2}) to the infinite-dimensional setting.  The analytic difficulties of this step --- one must evaluate an infinite-dimensional determinant and many infinite-dimensional traces --- are real, but we will suppose that they are dealt with in ways that do not interfere with the naive diagrammatic manipultations we use in \cref{section3}.
\end{enumerate}
Steps 1--3 define a formal power series $\Uu_\gamma(\phi)$ in $\hbar$, which we think of as the part of \cref{FeynmanIntegral} supported near $\gamma$.

The vertices in the Feynman diagrams in step 3 above are well-defined only once we have chosen a linearization in step 2.  Thus a priori the formal path integral depends on this linearization.  Nevertheless, in this paper we outline a proof of the following theorem:
\begin{thm} \label{MainThm}
  Given classical data $P \to X$, $\Aa: \Gamma(P \to X) \to \RR$, $\dVol \in \Gamma(\Ddd_{\fiber} \to P)$, $\phi \in \Gamma(\partial P \to \partial X)$, and a critical section $\gamma \in \Gamma_\phi(P \to X)$, the formal path integral $\Uu_\gamma(\phi)$ is independent of the choices in step 2 above.
\end{thm}
We call this paper an ``outline'' of a proof because when the analytic difficulties in step 3 are present, more work must be done to confirm that our diagrammatic manipulations in \cref{section4} are valid.  But our outline is an honest proof when these difficulties are absent.  This includes finite-dimensional integrals ($0$-dimensional QFT), when the argument is guaranteed to succeed because \cref{FeynmanIntegral} has an analytic definition, and quantum mechanics on a semi-Riemannian manifold (a certain $1$-dimensional QFT) \cite{meshort,melong}.

\section{Review of Feynman diagrams} \label{section2}

Let $M$ be a finite-dimensional compact smooth manifold with a chosen volume form $\dVol \in \Gamma(\Ddd_M \to M)$, and let $f: M \to \RR$ be a smooth function with isolated critical points.  Recall the ``big-$O$'' notation: if $x,y,z$ are expressions in $\hbar$, then ``$x = y + O(z)$'' means that $(x-y)/z$ is bounded as $\hbar \to 0$.  We also write ``$O(\hbar^\infty)$'' for ``$O(\hbar^n)$ for all $n$.''  Then \cite{EZ2007}:
\begin{equation} \label{asymp1}
  \int_M \exp\biggl( \frac i \hbar f\biggr)\dVol = \sum_{\substack{c \text{ a critical}\\ \text{point of }f}} \;\;\int_{\substack{\text{small nbhd of $c$}}} \limits  \exp\biggl( \frac i \hbar f\biggr)\dVol  \;\; +\;\; O(\hbar^\infty)
\end{equation}

Recall that for $n\geq 2$, if $M$ is a smooth manifold with no extra structure, the $n$th derivative $f^{(n)}$ of a smooth function $f: M \to \RR$ is not a well-defined coordinate-indendent object.  Rather, under changes of coordinates $f^{(n)}$ transforms in some affine way: there is a ``shift'' that depends linearly on $f^{(1)},\dots, f^{(n-1)}$.  In particular, the differential $f^{(1)} = \!\d f$ is always well-defined as a section of $\T^*M$, and if $c$ is a critical point of $f$ the \define{Hessian} $f^{(2)}(c)$ is well-defined as a symmetric element of $(\T^*_cM)^{\otimes 2}$.  A critical point $c$ is \define{nondegenerate} if the Hessian, thought of as a map $f^{(2)}(c): \T_cM \to \T^*_cM$, has trivial kernel.  The volume form $\dVol$ determines a \define{determinant} $\det_{\dVol}f^{(2)} = \bigl(f^{(2)}\bigr)^{\wedge \dim M} \cdot \bigl((\dVol)^{-1}\bigr)^{\otimes 2} \in \RR$, and it is non-zero if $c$ is nondegenerate.  The \define{signature} of $f^{(2)}$ is $\sign(f^{(2)}) = \dim_+ - \dim_-$, where $\dim_\pm$ is the dimension of the largest subspace of $\T_cM$ on which $\pm f^{(2)}(c) : (\T^*_cM)^{\otimes 2} \to \RR$ is positive-definite.  The leading asymptotics of \cref{asymp1} near nondegenerate critical points are:
\begin{equation} \label{asymp2}
  \int_{\text{small nbhd of }c} \limits e^{-(i\hbar)^{-1} f}\dVol = (2\pi \hbar)^{\dim M / 2} e^{\sign(f^{(2)})i\pi / 4}  e^{-(i\hbar)^{-1} f(c)} \bigl| \det f^{(2)}\bigr|^{-1/2} \bigl(1 + O(\hbar)\bigr)
\end{equation}
The reader should compare \cref{asymp2} with the standard formulas for Gaussian integrals.

\Cref{asymp2} is the farthest into the asymptotics of \cref{asymp1} that we can go explicitly without picking local coordinates on $M$.  For a sufficiently small neighborhood $\Oo$ of $c \in M$, by \cite{Moser1965} we can find coordinates $m: \Oo \to \RR^{\dim M}$ so that the volume form $\dVol$ is the pull back of the canonical volume form on $\RR^{\dim M}$.  With these coordinates, we expand $f$ in Taylor series around $c$: $f(c+x) = \sum_{n=0}^\infty f^{(n)}(c) \cdot x^{\otimes n}/n! + O(x^\infty)$, where $x\in \T_cM$ and $\cdot$ is the pairing between $f^{(n)}(c) : (\T_cM)^{\otimes n} \to \RR$ and $x^{\otimes n} \in (\T_cM)^{\otimes n}$.

Recall Penrose's graphical language \cite{Penrose1971} (made precise in \cite{JoyalStreet1}; it is implicit in Feynman's and Dyson's work on path integrals \cite{Feynman1948,Feynman1949a,Feynman1949b,Dyson1949a,Dyson1949b}).  Our convention will be to read diagrams from top to bottom.  We let a solid edge denote $\T_cM = \RR^{\dim M}$, and declare the following \define{Feynman rules}:
\begin{gather} \label{Frule1f}
  \begin{tikzpicture}[baseline=(X),green!50!black]
    \path node[dot] (O) {} ++(0pt,4pt) coordinate (X);
    \draw (O) -- ++(-15pt,15pt) +(0,3pt) node[anchor=base,text=black] {$\scriptstyle x_1$};
    \draw (O) -- ++(-6pt,15pt) +(0,3pt) node[anchor=base,text=black] {$\scriptstyle x_2$};
    \draw (O) -- ++(15pt,15pt) +(0,3pt) node[anchor=base,text=black] {$\scriptstyle x_n$};
    \path (O) ++(4pt,13pt) node[text=black] {$\scriptstyle \ldots$};
  \end{tikzpicture} 
   = -f^{(n)}(c)\cdot (x_1 \otimes \dots \otimes x_n) \\
   \begin{tikzpicture}[baseline=(X),green!50!black]
    \path coordinate (A) ++(0pt,4pt) coordinate (X);
    \path (A);
    \path (A) ++(20pt,0) coordinate (B);
    \draw (A) .. controls +(0,15pt) and +(0,15pt) .. (B);
  \end{tikzpicture}
   = \bigl(f^{(2)}(c)\bigr)^{-1}, \text{ i.e.\ }
   \begin{tikzpicture}[baseline=(X),green!50!black]
    \path node[dot] (O) {} ++(0pt,4pt) coordinate (X);
    \draw (O) -- ++(-10pt,18pt);
    \draw (O) .. controls +(10pt,15pt) and +(0,15pt) .. ++(15pt,0);
  \end{tikzpicture}   
  = -
  \begin{tikzpicture}[baseline=(X),green!50!black]
    \path coordinate (O) +(0pt,4pt) coordinate (X) +(-10pt,20pt) coordinate (t) +(15pt,0) coordinate (b);
    \draw (t) .. controls +(0pt,-8pt) and +(0pt,8pt) .. (b);
  \end{tikzpicture}   \label{Frule2f}
\end{gather}

A Feynman diagram is a combinatorial graph $G$ (it may be empty, disconnected, etc.).  More precisely, a \define{partition} of a finite set $N$ is a finite collection $S$ of pairwise-disjoint nonempty\footnote{We will never need zero-valent vertices, and use this (standard) definition to be compatible with \cref{FdB}, but one can easily modify the definition if zero-valent vertices are desired.} subsets of $N$ such that $\bigcup_{s\in S}s = N$, and a \define{Feynman diagram} is a finite collection $H$ of ``half-edges'' along with a finite partition $V$ of $H$ into blocks (the ``vertices'') and a finite partition $E$ of $H$ into blocks of size $2$ (the ``edges'').  An \define{isomorphism} of Feynman diagrams $\phi: (H,E,V) \to (H',E',V')$ is a bijection $H \to H'$ that induces bijections $E \to E'$ and $V \to V'$.  For a Feynman diagram $G$, we will write $\left| \Aut G\right|$ for the number of isomorphisms $G \to G$.\footnote{A more standard definition of a combinatorial graph is a finite set $V$ of ``vertices'' and an ``adjacency matrix'' $V \times V \to \NN$.  We prefer the definition we have given in terms of half edges because we want its corresponding notion of ``number of automorphisms.''  In particular, if $G$ is the diagram with one (bivalent) vertex and one edge connecting this vertex to itself, then we want $\left|\Aut G\right| = 2$.}  The \define{Euler characteristic} of a Feynman diagram is $\chi(H,E,V) = \left|V\right| - \left|E\right|$.  By placing a weight on each vertex and a balloon on each edge, we can \define{evaluate} each Feynman diagram in terms of the Feynman rules of \cref{Frule1f,Frule2f}; we will write $\ev(G) \in \RR$ for this evaluation.

Then, in terms of the coordinates $m: \Oo \to \T_cM$ and the above graphical language, we can fully describe the asymptotics of \cref{asymp2}:
\begin{multline} \label{asymp3}
  \int_{\Oo} e^{-(i\hbar)^{-1} f}\dVol = (2\pi \hbar)^{\dim M / 2} \times e^{\sign(f^{(2)})i\pi / 4} \times e^{-(i\hbar)^{-1} f(c)} \times \bigl| \det f^{(2)}\bigr|^{-1/2} \times \mbox{} \\
  \times \sum_{\substack{\text{Feynman diagrams }G \\ \text{with all vertices trivalent or higher}}} \frac{(i\hbar)^{-\chi(G)}\ev(G)}{\left|\Aut G\right|} \;\; + \;\; O(\hbar^\infty)
\end{multline}
For details, see e.g.\ \cite{Polyak2005}.  \Cref{asymp3} is equivalent to the following expression, which some readers might find more familiar \cite{EZ2007}, and which also requires coordinates to be defined:
\begin{multline*}
   \int_{\Oo} e^{-(i\hbar)^{-1} f}\dVol = \; O(\hbar^\infty) \; + \; (2\pi \hbar)^{\dim M / 2} \times e^{\sign(f^{(2)})i\pi / 4} \times \bigl| \det f^{(2)}(c)\bigr|^{-1/2} \times \mbox{} \\
   \mbox{} \times 
  \left. \exp \left( -\frac{i\hbar}2 \sum_{j,k = 1}^{\dim M} \bigl((f^{(2)})^{-1}\bigr)^{jk} \frac{\partial}{\partial m^j}\frac{\partial}{\partial m^k} \right)
  \exp \left(- (i\hbar)^{-1}\left( f(m) -  f^{(2)}(c) \cdot \frac{(m-c)^{\otimes 2}}2\right)\right)\right|_{m=c}
\end{multline*}

The idea of the formal path integral is to translate \cref{asymp3} to the infinite-dimensional setting of \cref{FeynmanIntegral}.  One must work to make sense of the dimension, Morse index, and determinant terms, but all the definitions that we know do not depend on the choice of coordinates in step 2 of \cref{section1}.  For the purposes of this note, we focus on the sum of diagrams:
\begin{defn} \label{MainDefn}
  Let $P \to X$ be a bundle, $\Aa: \Gamma(P\to X) \to \RR$ and action, $\phi \in \Gamma(\partial P \to \partial X)$ a Dirichlet boundary condition, and $\gamma \in \Gamma_\phi(P \to X)$ a critical section.  
  Pick a volume-preserving identification of an open neighborhood of the image of $\gamma$ with an open subbundle of $\T_\gamma\bigl(\Gamma_\phi(P \to X)\bigr)$.
    Then the \define{formal path integral} is the formal power series:
  \begin{equation} \label{formalintdefn}
    \Uu_\gamma(\phi) = \sum_{\substack{\text{Feynman diagrams }G \\ \text{with all vertices trivalent or higher}}} \frac{(i\hbar)^{-\chi(G)}\ev(G)}{\left|\Aut G\right|}
  \end{equation}
  The Feynman diagrams are evaluated via the following rules:
  \begin{gather} \label{Frule1}
  \begin{tikzpicture}[baseline=(X),green!50!black]
    \path node[dot] (O) {} ++(0pt,4pt) coordinate (X);
    \draw (O) -- ++(-15pt,15pt);
    \draw (O) -- ++(-6pt,15pt);
    \draw (O) -- ++(15pt,15pt);
    \path (O) ++(4pt,13pt) node[text=black] {$\scriptstyle \ldots$};
  \end{tikzpicture} 
   = -\Aa^{(n)}(\gamma) : \T_\gamma\bigl(\Gamma_\phi(P \to X)\bigr)^{\otimes n} \to \RR \\
   \begin{tikzpicture}[baseline=(X),green!50!black]
    \path coordinate (A) ++(0pt,4pt) coordinate (X);
    \path (A);
    \path (A) ++(20pt,0) coordinate (B);
    \draw (A) .. controls +(0,15pt) and +(0,15pt) .. (B);
  \end{tikzpicture}
   = \bigl((\Aa|_{\Gamma_\phi(P\to X)})^{(2)}(\gamma)\bigr)^{-1} \in \T_\gamma\bigl(\Gamma_\phi(P \to X)\bigr) \,\hat\otimes\, \T_\gamma \bigl(\Gamma_\phi(P \to X)\bigr)   \label{Frule2}
\end{gather}
\end{defn}
In \cref{Frule2}, the symbol $\hat\otimes$ denotes that the edge is valued in some completion of the tensor square, not in the algebraic tensor square itself.  We will ignore the problem that when evaluating diagrams in \cref{formalintdefn} one must compute infinite-dimensional traces, which usually diverge.

\section{On volume-preserving maps of bundles}  \label{section3}

By construction, $\Uu_\gamma(\phi)$ depends only on the value of the action $\Aa$ on a small open neighborhood of the image of the critical field $\gamma$.  Thus to prove \cref{MainThm}, it suffices to consider a vector bundle $\pi: Q \to X$, an open subbundle $P \subseteq Q$, and a volume-preserving but non-linear bundle map $f: P \to Q$.  A subset of a vector space $V$ is \define{star-shaped} if it is nonempty and for each $v\in V$ and for each $s\in [0,1]$, we have $sv \in V$; by shrinking $P$, we can suppose that the fibers $P_x = P \cap \pi^{-1}(x)$ are star-shaped for each $x\in X$.

We establish some notation.  $\tilde P = f(P)$ is the image bundle over $X$.  For each $x\in X$, $f_x: P_x \to \tilde P_x$ is the restriction of $f$ to the fibers over $x$, and we also use the letter $f$ for the induced map $\Gamma(P\to X) \to \Gamma(\tilde P \to X)$.  The initial action is $\Aa : \Gamma(P\to X) \to \RR$, and $\tilde \Aa = \Aa \circ f^{-1} : \Gamma(\tilde P \to X) \to \RR$ is the pushed-forward action.  Given $\phi \in \Gamma(\partial Q \to \partial X)$, let $\tilde \phi = f\circ \phi$.  Let $\gamma(\phi)$ be a choice of classical field extending $\phi$, and $\tilde\gamma(\tilde \phi) = f\circ \gamma(\phi)$.  In particular, $\tilde\gamma$ is a classical path for $\tilde \Aa$ extending $\tilde\phi$.   For fixed $\phi \in \Gamma(\partial P \to \partial X)$, our goal is to compare $\Uu_\gamma(\phi)$, the formal path integral determined by the action $\Aa$ and its classical path $\gamma(\phi)$, and $\tilde \Uu_{\tilde\gamma}(\tilde \phi)$, the formal path integral for $(\tilde\Aa,\tilde\gamma)$.

By construction, the formal path integral of \cref{MainDefn} is invariant under affine changes of coordinates.\footnote{A more honest definition of the formal path integral would come to grips with the determinant factor of \cref{asymp3}, and under any reasonable definition, this determinant is invariant under volume-preserving affine changes of coordinates, but not under non-volume-preserving maps.}  By applying a suitable translation, we can assume that $f(0) = 0$.  Then for each $x\in X$, the differential $\!\d f_x|_0 = (f_x)^{(1)}(0)$ is a linear volume-preserving map: $(f_x)^{(1)}(0) \in \GL(\T^{\fiber}_{(0,x)}Q) = \GL(Q_x)$ with $\det\bigl((f_x)^{(1)}(0)\bigr) = \pm 1$.  By composing $f: P \to \tilde P$ with $\{x \mapsto \bigl((f_x)^{(1)}(0)\bigr)^{-1} \} \in \Gamma( \GL(Q) \to X \bigr)$, we can suppose that for each $x \in X$, $(f_x)^{(1)}(0) = 1 \in \GL(Q_x)$.   It follows that $f: P \to Q$ is not just volume- but orientation-preserving on each fiber.  We say that a smooth (but non-linear) bundle map $f: P \to Q$ is \define{locally volume- and orientation-preserving} if for each $x\in X$ and for each $p\in P_x$, $\det\bigl((f_x)^{(1)}(p)\bigr) = 1$.  A locally volume-preserving map might fail to be globally so if it is not injective; by shrinking $P\subseteq Q$ to a smaller open neighborhood, we can always avoid this, and in any case our proof of \cref{MainThm} never relies on global properties of $f$.

We will henceforth simplify the notation: $f^{(n)}$ refers to the $n$th derivative of $f$ in the fiber direction, so that for example $f^{(1)}(q) = \bigl(f_{\pi(q)}\bigr)^{(1)}(q)$. 

\begin{lemma} \label{MainLemma}
  Let $\pi: Q \to X$ be a vector bundle and $P \subseteq Q$ an open subbundle such that for each $x\in X$, the fiber $P_x = P \cap \pi^{-1}(x)$ is star-shaped.  Then any locally volume- and orientation-preserving map $f: P \to Q$ with $f(0) = 0$ and $f^{(1)}(0) = 1$ is smoothly homotopic among such maps to the identity map $\id: P \mono Q$.
\end{lemma}

\begin{proof}
  For each $q \in P$, $g(q) = f^{(1)}(q)$ is a linear map $Q_{\pi(q)} = \T_q^{\fiber}Q \to \T_{f(q)}^{\fiber}Q = Q_{\pi(q)}$.  Its derivative, a linear map $g^{(1)}(q) : \bigl(Q_{\pi(q)}\bigr)^{\otimes 2} \to Q_{\pi(q)}$, is naturally defined.  This map satisfies the following three conditions:
  \begin{enumerate}
    \item $g^{(1)}(q)$ is symmetric for each $q\in P$.
    \item $\det g(q) = 1$ for each $g\in P$.
    \item $g(0) = 1: Q_{\pi(q)} \to Q_{\pi(q)}$ for each $x\in X$.
  \end{enumerate}
  Conversely, since each fiber $P_x$ is star-shaped and hence contractible, any smooth family of linear maps $g(q) : Q_{\pi(q)} \to Q_{\pi(q)}$ satisfying condition 1 above determines a unique function $f: P \to Q$ with $f(0) = 0$ and $f^{(1)}(q) = g(q)$, and if $g$ satisfies conditions 2 and 3 then $f$ is locally volume- and orientation-preserving and $f^{(1)}(0) = 1$.
  
  Thus, let $f: P \to Q$ be locally volume- and orientation-preserving with $f^{(1)}(0) = 1$ and $f(0) = 0$, and for each $q\in P$ and each $s\in [0,1]$ let $G(s,q) = f^{(1)}(sq) : Q_{\pi(q)} \to Q_{\pi(q)}$.  Then $g = G(s,-)$ satisfies conditions 1--3 above.  Let $F(s,-)$ be the corresponding map $P \to Q$.  Then $F(1,-) = f$, whereas $F(0,-) = \id$ as its derivative is $F(0,-)^{(1)} = 1$.  By ``differentiating under the integral,'' $F$ is obviously smooth in $s$.  Thus $F$ is the desired homotopy.
\end{proof}

\section{A diagrammatic calculation} \label{section4}

In this section we complete the proof of \cref{MainThm}.  As in the previous section, we have a vector bundle $Q \to X$, an open star-shaped subbundle $P \to X$, an action $\Aa: \Gamma(P \to X) \to \RR$, a volume-preserving smooth bundle map $f: P \to Q$ with image $\tilde P$, and the pushed-forward action $\tilde\Aa = \Aa \circ f^{-1} \in \Gamma(\tilde P \to X)$, and our goal is to compare $\Uu_\gamma(\phi)$ and $\tilde\Uu_{\tilde \gamma}(\tilde\phi)$.  By \cref{MainLemma}, we may assume that $f$ is an infinitesimal transformation: $f(q) = q + e(q) + O(e^2)$, where $e \in \Gamma(\T^{\fiber}Q \to Q)$ is a fiberwise vector field.  Then $P = \tilde P$.

Recall the generalized chain rule by Fa\`a di Bruno \cite{Hardy2006}:
\begin{equation} \label{FdB}
  \bigl(\Aa \circ f^{-1}\bigr)^{(n)}\cdot \bigl(\xi_1 \otimes \xi_n\bigr) = \sum_{\text{partitions $S$ of }\{1,\dots,n\}} \bigl(\Aa^{(|S|)}\circ f^{-1}\bigr) \cdot \bigotimes_{s \in S} \Bigl( \bigl(f^{-1}\bigr)^{(|s|)} \cdot \bigotimes_{j \in s} \xi_j \Bigr)
\end{equation}
The partition determines how to contract (abbreviated ``$\cdot$'') the tensors in \cref{FdB}.

The Penrose graphical language expresses \cref{FdB} well.  As in \cref{Frule1,Frule2}, let:
\begin{equation} \label{Frule}
  \begin{tikzpicture}[baseline=(X),green!50!black]
    \path node[dot] (O) {} ++(0pt,4pt) coordinate (X);
    \draw (O) -- ++(-15pt,15pt) +(0,3pt) node[anchor=base,text=black] {$\scriptstyle \xi_1$};
    \draw (O) -- ++(-6pt,15pt) +(0,3pt) node[anchor=base,text=black] {$\scriptstyle \xi_2$};
    \draw (O) -- ++(15pt,15pt) +(0,3pt) node[anchor=base,text=black] {$\scriptstyle \xi_n$};
    \path (O) ++(4pt,13pt) node[text=black] {$\scriptstyle \ldots$};
  \end{tikzpicture} = - \Aa^{(n)}(\gamma) \cdot (\xi_1 \otimes \dots \otimes \xi_n) 
  \quad\quad\quad\quad
  \begin{tikzpicture}[baseline=(X),green!50!black]
    \path coordinate (A) ++(0pt,4pt) coordinate (X);
    \path (A);
    \path (A) ++(20pt,0) coordinate (B);
    \draw (A) .. controls +(0,15pt) and +(0,15pt) .. (B);
  \end{tikzpicture}
   = \bigl(\Aa^{(2)}\bigr)^{-1}
\end{equation}
Here $\xi_k \in \T_\gamma\bigl(\Gamma(Q \to X)\bigr)$.  We also introduce the following notation:
\begin{gather} \label{Prule1}
  \begin{tikzpicture}[baseline=(X),green!50!black]
    \path node[circle,inner sep=0pt,draw] (O) {$\scriptstyle \sim$} ++(0pt,4pt) coordinate (X);
    \draw (O) -- ++(-15pt,15pt) +(0,3pt) node[anchor=base,text=black] {$\scriptstyle \xi_1$};
    \draw (O) -- ++(-6pt,15pt) +(0,3pt) node[anchor=base,text=black] {$\scriptstyle \xi_2$};
    \draw (O) -- ++(15pt,15pt) +(0,3pt) node[anchor=base,text=black] {$\scriptstyle \xi_n$};
    \path (O) ++(4pt,13pt) node[text=black] {$\scriptstyle \ldots$};
  \end{tikzpicture} = - \tilde\Aa^{(n)}(\tilde\gamma) \cdot (\xi_1 \otimes \dots \otimes \xi_n)
  \quad\quad\quad\quad
  \begin{tikzpicture}[baseline=(X),green!50!black]
    \path coordinate (A) ++(0pt,4pt) coordinate (X);
    \path (A);
    \path (A) ++(20pt,0) coordinate (B);
    \draw (A) .. controls +(0,15pt) and +(0,15pt) .. node[fill=white,circle,inner sep=0pt,draw] {$\scriptstyle \sim$} (B);
  \end{tikzpicture} 
   = \bigl(\tilde\Aa^{(2)}\bigr)^{-1} \\
  \begin{tikzpicture}[baseline=(X),green!50!black]
    \path node[draw,rectangle,text=black,inner sep=1pt] (O) {$\scriptstyle f^{-1}$} ++(0pt,0pt) coordinate (X);
    \draw (O) -- ++(-15pt,15pt) +(0,3pt) node[anchor=base,text=black] {$\scriptstyle \xi_1$};
    \draw (O) -- ++(-6pt,15pt) +(0,3pt) node[anchor=base,text=black] {$\scriptstyle \xi_2$};
    \draw (O) -- ++(15pt,15pt) +(0,3pt) node[anchor=base,text=black] {$\scriptstyle \xi_n$};
    \path (O) ++(4pt,13pt) node[text=black] {$\scriptstyle \ldots$};
    \draw (O) -- ++(0pt,-15pt);
  \end{tikzpicture} = \Bigl\{ x \mapsto \bigl(f^{-1}\bigr)^{(n)}(\gamma(x)) \cdot \bigl(\xi_1(x) \otimes \dots \otimes \xi_n(x)\bigr) \Bigr\} \in \T_\gamma\bigl(\Gamma(Q \to X)\bigr) \label{Prule2}
\end{gather}
In the right-hand side of \cref{Prule2}, the derivatives of $f^{-1}$ are taken in the fiber direction, and this equation defines the $n$th derivative of $f^{-1}$ understood as a map $\Gamma(\tilde P\to X) \to \Gamma(P\to X)$; it is this infinite-dimensional derivative that is meant in \cref{FdB}.

Then \cref{FdB} reads:
\begin{equation}
  \begin{tikzpicture}[baseline=(X),green!50!black]
    \path node[circle,inner sep=0pt,draw] (O) {$\scriptstyle \sim$} ++(0pt,10pt) coordinate (X);
    \draw (O) -- ++(-15pt,30pt);
    \draw (O) -- ++(-6pt,30pt);
    \draw (O) -- ++(15pt,30pt);
    \path (O) ++(4pt,26pt) node[text=black] {$\scriptstyle \ldots$};
  \end{tikzpicture} =
  \sum 
  \begin{tikzpicture}[baseline=(X),green!50!black]
    \path node[dot] (O) {} ++(0pt,10pt) coordinate (X);
    \path (O)  ++(-15pt,15pt)  node[draw,rectangle,text=black,inner sep=1pt] (a1) {$\scriptstyle f^{-1}$};
    \path (O)  ++(15pt,15pt)  node[draw,rectangle,text=black,inner sep=1pt] (a2) {$\scriptstyle f^{-1}$};
    \draw (O) -- (a1);
    \draw (O) -- (a2);
    \path (O) ++(0pt,13pt) node[text=black] {$\scriptstyle \ldots$};
    \draw (a1) -- ++(0pt,15pt);
    \draw (a1) -- ++(4pt,15pt);
    \draw (a2) -- ++(0pt,15pt);
    \path (a1) ++(8pt,15pt) coordinate (b1);
    \path (a2) ++(-4pt,15pt) coordinate (b2);
    \path (a1) ++(-4pt,15pt) coordinate (b4);
    \draw (a2) -- ++(4pt,15pt);
    \draw (a1) -- (b2);
    \draw (a2) -- (b1);
    \draw (a2) -- (b4);
    \path (O) ++(2pt,29pt) node[text=black] {$\scriptstyle \ldots$};
  \end{tikzpicture} 
\end{equation}
The sum ranges over isomorphism classes of diagrams with ordered leaves but unordered \tikz[baseline=(a.base)] \node[rectangle,draw,green!50!black,text=black,anchor=base,inner sep=0pt] (a) {$f^{-1}$}; vertices.  The \tikz \node[dot,green!50!black] {}; vertex can be of arbitrary valence (non-zero, if the left-hand-side has non-zero valence), and each \tikz[baseline=(a.base)] \node[rectangle,draw,green!50!black,text=black,anchor=base,inner sep=0pt] (a) {$f^{-1}$}; vertex has one output strand and at least one input strand.  The \tikz \node[green!50!black,circle,inner sep=0pt,draw] (O) {$\scriptstyle \sim$}; vertex on the left-hand side and the  \tikz[baseline=(a.base)] \node[rectangle,draw,green!50!black,text=black,anchor=base,inner sep=0pt] (a) {$f^{-1}$}; vertices on the right hand side are evaluated at $\tilde\gamma$, and the \tikz \node[dot,green!50!black] {}; vertex on the right hand side is evaluated at $\gamma = f^{-1}\circ \tilde\gamma$.

But recall that $f(q) = q + e(q)$ and that we are free to work to $O(e^2)$ accuracy.  Then $f^{-1}(q) = q - e(q) + O(e^2)$, $\bigl(f^{-1}\bigr)^{(1)} = {\id} - e^{(1)}$, and $\bigl(f^{-1}\bigr)^{(n)} = -e^{(n)}$ for $n \geq 2$.  Writing 
\begin{equation}
\begin{tikzpicture}[baseline=(X),green!50!black]
    \path node[draw,circle,text=black,inner sep=1pt] (O) {$e$} ++(0pt,0pt) coordinate (X);
    \draw (O) -- ++(-15pt,15pt) +(0,3pt) node[anchor=base,text=black] {$\scriptstyle \xi_1$};
    \draw (O) -- ++(-6pt,15pt) +(0,3pt) node[anchor=base,text=black] {$\scriptstyle \xi_2$};
    \draw (O) -- ++(15pt,15pt) +(0,3pt) node[anchor=base,text=black] {$\scriptstyle \xi_n$};
    \path (O) ++(4pt,13pt) node[text=black] {$\scriptstyle \ldots$};
    \draw (O) -- ++(0pt,-15pt);
  \end{tikzpicture} = e^{(n)}(\gamma(x)) \cdot \bigl(\xi_1(x) \otimes \dots \otimes \xi_n(x)\bigr),
\end{equation}
we have:
\begin{equation}
  \begin{tikzpicture}[baseline=(X),green!50!black]
    \path node[circle,inner sep=0pt,draw] (O) {$\scriptstyle \sim$} ++(0pt,10pt) coordinate (X);
    \draw (O) -- ++(-15pt,30pt);
    \draw (O) -- ++(-6pt,30pt);
    \draw (O) -- ++(15pt,30pt);
    \path (O) ++(4pt,26pt) node[text=black] {$\scriptstyle \ldots$};
  \end{tikzpicture} =
  \begin{tikzpicture}[baseline=(X),green!50!black]
    \path node[dot] (O) {} ++(0pt,10pt) coordinate (X);
    \draw (O) -- ++(-15pt,30pt);
    \draw (O) -- ++(-6pt,30pt);
    \draw (O) -- ++(15pt,30pt);
    \path (O) ++(4pt,26pt) node[text=black] {$\scriptstyle \ldots$};
  \end{tikzpicture}  -
  \sum 
  \begin{tikzpicture}[baseline=(X),green!50!black]
    \path node[dot] (O) {} ++(0pt,10pt) coordinate (X);
    \path (O)  ++(-15pt,15pt) coordinate (a1);
    \path (O)  ++(15pt,15pt)  node[draw,circle,text=black,inner sep=1pt] (a2) {$e$};
    \draw (O) -- (a2);
    \draw (a1) ++(0pt,15pt) coordinate (c1);
    \draw (O) -- (c1);
    \draw (a1) ++(4pt,15pt) coordinate (c2);
    \draw (O) -- (c2);
    \draw (a2) -- ++(0pt,15pt);
    \path (a1) ++(8pt,15pt) coordinate (b1);
    \path (a2) ++(-4pt,15pt) coordinate (b2);
    \path (a1) ++(-4pt,15pt) coordinate (b4);
    \draw (a2) -- ++(4pt,15pt);
    \draw (O) -- (b2);
    \draw (a2) -- (b1);
    \draw (a2) -- (b4);
    \path (O) ++(2pt,29pt) node[text=black] {$\scriptstyle \ldots$};
  \end{tikzpicture} + O(e^2) \label{proof1}
\end{equation}
The sum in \cref{proof1} ranges over diagrams with precisely one \tikz[baseline=(e.base)] \node[draw,circle,inner sep=1pt,anchor=base,green!50!black,text=black] (e) {$e$}; vertex, which must have at least one input string.  The \tikz \node[dot,green!50!black] {}; vertex on the right-hand side can be of arbitrary valence.  By \cref{Frule2,Prule1}, we also have:
\begin{equation}
  \begin{tikzpicture}[baseline=(X),green!50!black]
    \path coordinate (A) ++(0pt,4pt) coordinate (X);
    \path (A);
    \path (A) ++(20pt,0) coordinate (B);
    \draw (A) .. controls +(0,15pt) and +(0,15pt) .. node[fill=white,circle,inner sep=0pt,draw] {$\scriptstyle \sim$} (B);
  \end{tikzpicture}  
  =
  \begin{tikzpicture}[baseline=(X)]
    \path coordinate (A) ++(0pt,2pt) coordinate (X);
    \path (A) ++(20pt,0) coordinate (B);
    \draw[green!50!black] (A) .. controls +(0,15pt) and +(0,15pt) .. (B);
  \end{tikzpicture}
  +
    \begin{tikzpicture}[baseline=(X)]
    \path coordinate (A) ++(0pt,2pt) coordinate (X);
    \path (A) ++(20pt,0) coordinate (B);
    \draw[green!50!black] (A) .. controls +(0,15pt) and +(0,15pt) .. (B);
    \path (B) node[draw,circle,green!50!black,text=black,inner sep=1pt,fill=white] (E) {$e$};
    \draw[green!50!black] (E) -- ++(0pt,-10pt);
  \end{tikzpicture}
  +
    \begin{tikzpicture}[baseline=(X)]
    \path coordinate (A) ++(0pt,2pt) coordinate (X);
    \path (A) ++(20pt,0) coordinate (B);
    \draw[green!50!black] (A) .. controls +(0,15pt) and +(0,15pt) .. (B);
    \path (A) node[draw,circle,green!50!black,text=black,inner sep=1pt,fill=white] (E) {$e$};
    \draw[green!50!black] (E) -- ++(0pt,-10pt);
  \end{tikzpicture}
  + O(e^2) \label{proof2}
\end{equation}

Finally, consider the sum from \cref{asymp3} of trivalent-and-higher diagrams made from the ingredients in \cref{Prule1}.  By \cref{proof1,proof2}, this sum is equal to the sum of diagrams made from the ingredients in \cref{Frule} plus some extra diagrams with \tikz[baseline=(e.base)] \node[draw,circle,inner sep=1pt,anchor=base,green!50!black,text=black] (e) {$e$}; vertices; to order $O(e^2)$, each extra diagram has exactly one such vertex.  But many of these extra diagrams cancel.  In particular, the \tikz[baseline=(e.base)] \node[draw,circle,inner sep=1pt,anchor=base,green!50!black,text=black] (e) {$e$}; vertices appear with opposite signs in \cref{proof1,proof2}, and so we can cancel any extra diagram coming from \cref{proof2} with an extra diagram coming from \cref{proof1} in which the \tikz[baseline=(e.base)] \node[draw,circle,inner sep=1pt,anchor=base,green!50!black,text=black] (e) {$e$}; vertex has precisely one input string.  (One must check whenever we claim some diagrams cancel that those diagrams appear in \cref{asymp3} with the same multiplicities / symmetry factors $\left| \Aut G\right|$, but this follows from some very simple combinatorics.)

In \cref{proof1}, the \tikz \node[dot,green!50!black] {}; vertex can have arbitrary valence, and in particular it can have valence two, provided the corresponding \tikz[baseline=(e.base)] \node[draw,circle,inner sep=1pt,anchor=base,green!50!black,text=black] (e) {$e$}; vertex has at least two input strings, or it can have valence one, provided the \tikz[baseline=(e.base)] \node[draw,circle,inner sep=1pt,anchor=base,green!50!black,text=black] (e) {$e$}; vertex has at least three inputs.  But each input to the \tikz[baseline=(e.base)] \node[draw,circle,inner sep=1pt,anchor=base,green!50!black,text=black] (e) {$e$}; vertex will be one end of an edge from \cref{Frule2}, and in particular an element of $\T_\gamma\bigl(\Gamma_\phi(P\to X)\bigr)$, and so:
\begin{equation} \label{proof3}
  \begin{tikzpicture}[baseline=(X),green!50!black]
    \path node[dot] (O) {} ++(0pt,12pt) coordinate (X);
    \draw (O) -- ++(-10pt,18pt) node[draw,circle,text=black,inner sep=1pt,fill=white] (E) {$e$};
    \draw (O) .. controls +(10pt,15pt) and +(0,15pt) .. ++(15pt,0);
    \draw (E) -- ++(-10pt,15pt);
    \draw (E) -- ++(-4pt,15pt);
    \draw (E) -- ++(10pt,15pt);
    \path (E) ++(2pt,13pt) node[text=black] {$\scriptstyle \ldots$};
  \end{tikzpicture}  
  = -
  \begin{tikzpicture}[baseline=(X),green!50!black]
    \path node[white,dot] (O) {} +(0pt,12pt) coordinate (X) +(-10pt,20pt) coordinate (t) +(15pt,0) coordinate (b);
    \draw (t) .. controls +(0pt,-8pt) and +(0pt,8pt) .. (b);
    \path (O) ++(-10pt,18pt) node[draw,circle,text=black,inner sep=1pt,fill=white] (E) {$e$};
    \draw (E) -- ++(-10pt,15pt);
    \draw (E) -- ++(-4pt,15pt);
    \draw (E) -- ++(10pt,15pt);
    \path (E) ++(2pt,13pt) node[text=black] {$\scriptstyle \ldots$};
  \end{tikzpicture}
  \quad\quad\quad\quad
  \begin{tikzpicture}[baseline=(X),green!50!black]
    \path node[dot] (O) {} ++(0pt,12pt) coordinate (X);
    \draw (O) -- ++(0pt,18pt) node[draw,circle,text=black,inner sep=1pt,fill=white] (E) {$e$};
    \draw (E) -- ++(-10pt,15pt);
    \draw (E) -- ++(-4pt,15pt);
    \draw (E) -- ++(10pt,15pt);
    \path (E) ++(2pt,13pt) node[text=black] {$\scriptstyle \ldots$};
  \end{tikzpicture}  = 0
\end{equation}
The first part of \cref{proof3} is just the second part of \cref{Frule2f}, which applies in this context when the inputs are elements of $\T_\gamma\bigl(\Gamma_\phi(P\to X)\bigr) \subseteq \T_\gamma\bigl(\Gamma(P\to X)\bigr)$.  The second part of \cref{proof3} is the statement that $\gamma$ is a critical point for $\Aa|_{\Gamma_\phi(P\to X)}$, and again requires that the inputs be in $\T_\gamma\bigl(\Gamma_\phi(Q\to X)\bigr)$.  This certainly occurs: the \tikz[baseline=(e.base)] \node[draw,circle,inner sep=1pt,anchor=base,green!50!black,text=black] (e) {$e$}; vertices each have at least one input string, which is contracted with an edge $\bigl(\Aa^{(2)}\bigr)^{-1}$.
Thus we can cancel almost all the diagrams.

Indeed, the only diagrams left to cancel are those with a component of the form:
\begin{equation}
   -\begin{tikzpicture}[baseline=(X),green!50!black]
    \path node[dot] (O) {} ++(0pt,12pt) coordinate (X);
    \draw (O) -- ++(-10pt,18pt) node[draw,circle,text=black,inner sep=1pt,fill=white] (E) {$e$};
    \draw (O) .. controls +(15pt,15pt) and +(10pt,15pt) .. (E);
    \draw (E) -- ++(-10pt,15pt);
    \draw (E) -- ++(4pt,15pt);
    \path (E) ++(-3pt,13pt) node[text=black] {$\scriptstyle \ldots$};
  \end{tikzpicture}  , \text{ at least one external string}
\end{equation}
This is a ``trace'' of a derivative of $e$.  Indeed, by \cref{Frule2}, we have, formally:
\begin{equation} \label{proofn}
   -\begin{tikzpicture}[baseline=(E.base),green!50!black]
    \path node[dot] (O) {} ++(0pt,15pt) coordinate (X);
    \draw (O) -- ++(-10pt,18pt) node[draw,circle,text=black,inner sep=1pt,fill=white] (E) {$e$};
    \draw (O) .. controls +(15pt,15pt) and +(10pt,15pt) .. (E);
    \draw (E) -- ++(-10pt,15pt) +(0,3pt) node[anchor=base,text=black] {$\scriptstyle \xi_1$};
    \draw (E) -- ++(4pt,15pt) +(0,3pt) node[anchor=base,text=black] {$\scriptstyle \xi_n$};
    \path (E) ++(-3pt,13pt) node[text=black] {$\scriptstyle \ldots$};
  \end{tikzpicture}  
  = \sum_{x\in X \sminus \partial X}  \tr_{\T^{\fiber}_{\gamma(x)}Q} \Bigl\{ v \mapsto e^{(n)}(\gamma(x)) \cdot \bigl(\xi_1(x) \otimes \dots \otimes \xi_n(x) \otimes v \bigr) \Bigr\}
\end{equation}
The sum becomes an honest integral in theories with no ultraviolet divergences; c.f.\ \cite{melong}.  But $f = \id \!+ e: P \to Q$ is locally volume-preserving if and only if $\tr_{\T^{\fiber}_qQ} \bigl\{ v \mapsto e^{(1)}(q)\cdot v \bigr\} = 0$ for all $q$, and so the summands in \cref{proofn} all vanish.

Thus the sums of diagrams for the path integrals $\Uu,\tilde\Uu$ in the two coordinates systems agree up to $O(e^2)$ accuracy.  This completes the proof of \cref{MainThm}.  We remark that if $q \mapsto q + e(q)$ were not locally volume-preserving (up to a constant), the diagrams in \cref{proofn} would not vanish, and the value of the formal path integral would change.  This is just as expected from an ``integral.''

\section{Acknowledgments}

This note is part of a larger project \cite{meshort,melong} suggested by N.\ Reshetikhin, and it is with great pleasure that I thank him for providing support, advice, and ideas throughout.  I would also like to thank G.\ Kuperberg, G.\ Thompson, and C.\ Schommer-Pries for their clarifying discussions, some of which occurred on the on-line discussion forum \url{mathoverflow.net}.   This work was supported by the NSF grant DMS-0901431.

\bibliography{Edited}{}

\begin{thebibliography}{10}

\bibitem{BN1995}
Dror Bar-Natan.
\newblock Perturbative {C}hern-{S}imons theory.
\newblock {\em J. Knot Theory Ramifications}, 4(4):503--547, 1995.

\bibitem{Dyson1949a}
F.J. Dyson.
\newblock The radiation theories of {T}omonaga, {S}chwinger, and {F}eynman.
\newblock {\em Physical Review}, 75:486--502, 1948.

\bibitem{Dyson1949b}
F.J. Dyson.
\newblock The {S} matrix in quantum electrodynamics.
\newblock {\em Physical Review}, 75:1736--1755, 1949.

\bibitem{EZ2007}
L.C. Evans and M.~Zworski.
\newblock Lectures on semiclassical analysis.
\newblock Available at
  \url{http://math.berkeley.edu/~zworski/semiclassical.pdf}, 2007.

\bibitem{Feynman1948}
R.P. Feynman.
\newblock Space-time approach to non-relativistic quantum mechanics.
\newblock {\em Reviews of Modern Physics}, 20(2):367--387, 1948.

\bibitem{Feynman1949b}
R.P. Feynman.
\newblock Space-time approach to quantum electrodynamics.
\newblock {\em Physical Review}, 76:769--789, 1949.

\bibitem{Feynman1949a}
R.P. Feynman.
\newblock The theory of positrons.
\newblock {\em Physical Review}, 76:749--759, 1949.

\bibitem{FeynmanHibbs1965}
R.P. Feynman and A.R. Hibbs.
\newblock {\em Quantum mechanics and path integrals}.
\newblock International series in pure and applied physics. McGraw-Hill, New
  York, 1965.

\bibitem{Hardy2006}
M.~Hardy.
\newblock Combinatorics of partial derivatives.
\newblock {\em Electron. J. Combin.}, 13(1):Research Paper 1, 13 pp.
  (electronic), 2006.
\newblock Available at
  \url{http://www.combinatorics.org/Volume_13/Abstracts/v13i1r1.html}.

\bibitem{meshort}
T.~Johnson-Freyd.
\newblock Feynman-diagrammatic description of the asymptotics of the time
  evolution operator in quantum mechanics.
\newblock Available at \url{http://math.berkeley.edu/~theojf/QM1.pdf} or
  \href{http://arxiv.org/abs/1003.1156}{\tt arXiv:1003.1156 [math-ph]}, 2010.

\bibitem{melong}
T.~Johnson-Freyd.
\newblock The formal path integral and quantum mechanics.
\newblock In preparation, 2010.

\bibitem{JoyalStreet1}
A.~Joyal and R.~Street.
\newblock The geometry of tensor calculus. {I}.
\newblock {\em Adv. Math.}, 88(1):55--112, 1991.

\bibitem{Moser1965}
J.~Moser.
\newblock On the volume elements on a manifold.
\newblock {\em Trans. Amer. Math. Soc.}, 120:286--294, 1965.

\bibitem{Penrose1971}
R.~Penrose.
\newblock Applications of negative dimensional tensors.
\newblock In D.J.A. Welsh, editor, {\em Combinatorial mathematics and its
  applications}, pages 221--244, London, 1971. Mathematical Institute, Oxford,
  Academic Press.

\bibitem{Polyak2005}
M.~Polyak.
\newblock Feynman diagrams for pedestrians and mathematicians.
\newblock In M.~Lyubich and L.A. Takhtajan, editors, {\em Graphs and patterns
  in mathematics and theoretical physics}, volume~73 of {\em Proc. Sympos. Pure
  Math.}, pages 15--42, Providence, RI, 2005. Amer. Math. Soc.

\bibitem{KolyaShort}
N.~Reshetikhin.
\newblock Topological quantum field theory: 20 years later.
\newblock 2010.

\bibitem{Witten1989}
Edward Witten.
\newblock Quantum field theory and the {J}ones polynomial.
\newblock {\em Comm. Math. Phys.}, 121(3):351--399, 1989.

\end{thebibliography}
\bibliographystyle{plain}

\end{document}